\documentclass[a4paper,UKenglish,cleveref,autoref]{lipics-v2021}
\usepackage{amsthm}
\usepackage{setspace}

\pdfoutput=1 %uncomment to ensure pdflatex processing (mandatatory e.g. to submit to arXiv)
\hideLIPIcs  %uncomment to remove references to LIPIcs series (logo, DOI, ...), e.g. when preparing a pre-final version to be uploaded to arXiv or another public repository

\title{Interval-Memoized Backtracking on ZDDs for Fast Enumeration of All Lower Cost Solutions} %TODO Please add

\titlerunning{Interval-Memoized Backtracking on ZDDs for Fast Enumeration of All Lower Cost Solutions} %TODO optional, please use if title is longer than one line

\author{Shin-ichi Minato}{Kyoto University, Japan }{}{}{}%TODO mandatory, please use full name; only 1 author per \author macro; first two parameters are mandatory, other parameters can be empty. Please provide at least the name of the affiliation and the country. The full address is optional
\author{Mutsunori Banbara}{Nagoya University, Japan }{}{}{}
\author{Takashi Horiyama}{Hokkaido University, Japan }{}{}{}
\author{Jun Kawahara}{Kyoto University, Japan }{}{}{}
\author{Ichigaku Takigawa}{RIKEN Center for AIP, Japan / Hokkaido University, Japan}{}{}{}
\author{Yutaro Yamaguchi}{Osaka University, Japan }{}{}{}
\authorrunning{S.~Minato, M.~Banbara, T.~Horiyama, J.~Kawahara, I.~Takigawa, and Y.~Yamaguchi} %TODO mandatory. First: Use abbreviated first/middle names. Second (only in severe cases): Use first author plus 'et al.'

\Copyright{Shin-ichi Minato, Mutsunori Banbara, Takashi Horiyama, Jun Kawahara, Ichigaku Takigawa, and Yutaro Yamaguchi}  %TODO mandatory, please use full first names. LIPIcs license is "CC-BY";  http://creativecommons.org/licenses/by/3.0/

 %TODO mandatory: Please choose ACM 2012 classifications from https://dl.acm.org/ccs/ccs_flat.cfm 

 %TODO mandatory; please add comma-separated list of keywords

\category{} %optional, e.g. invited paper

\relatedversion{} %optional, e.g. full version hosted on arXiv, HAL, or other respository/website
\funding{This research is partly supported by JSPS KAKENHI Grant 20H00605 and 20H05964.}%optional, to capture a funding statement, which applies to all authors. Please enter author specific funding statements as fifth argument of the \author macro.

\nolinenumbers %uncomment to disable line numbering

\EventEditors{John Q. Open and Joan R. Access}
\EventNoEds{2}
\EventLongTitle{42nd Conference on Very Important Topics (CVIT 2016)}
\EventShortTitle{CVIT 2016}
\EventAcronym{CVIT}
\EventYear{2016}
\EventDate{December 24--27, 2016}
\EventLocation{Little Whinging, United Kingdom}
\EventLogo{}
\SeriesVolume{42}
\ArticleNo{23}
\newtheorem{theo}{Theorem}[section]

\begin{document}

\maketitle

\begin{abstract}
In this paper, we propose a fast method for exactly enumerating a very large number of all lower cost solutions for various combinatorial problems. Our method is based on backtracking for a given decision diagram which represents all the feasible solutions. The main idea is to memoize the intervals of cost bounds to avoid duplicate search in the backtracking process. In contrast to usual pseudo-polynomial-time dynamic programming approaches, the computation time of our method does not directly depend on the total cost values, but is bounded by the input and output size of the decision diagrams. Therefore, it can be much faster if the cost values are large but the input/output decision diagrams are well-compressed. We demonstrate its practical efficiency by comparing our method to current available enumeration methods: for nontrivial size instances of the Hamiltonian path problem, our method succeeded in exactly enumerating billions of all lower cost solutions in a few seconds, which was hundred or much more times faster. Our method can be regarded as a novel search algorithm which integrates the two classical techniques, branch-and-bound and dynamic programming. This method would have many applications in various fields, including operations research, data mining, statistical testing, hardware/software system design, etc.
\end{abstract}

\baselineskip = 4.4mm
\section{Introduction}
\label{sec:Introduction}
In many real-life situations, we are often faced with combinatorial optimization problems under linear cost functions, to find an optimal solution having the minimum (or maximum) total cost of items satisfying some constraints. This is one of the most popular types of classical problems, including shortest path, minimum spanning tree, traveling salesperson, knapsack, etc. For such problems, some powerful software tools, such as ILP solvers \cite{IBM-CPLEX2019,gurobi2020}, SAT-based CSP solvers \cite{Tamura09:Sugar,Sugar2018}, and BDD-based solvers \cite{Nishino2015:BDD-DPsearch} have been developed and successfully utilized for many practical problems.

Unfortunately, the optimal solution is not always the best decision in real-life situations since it is sometimes difficult to correctly model the problems which we really want to solve. Recent progress of computing resources enables us to not only generate one optimal solution but also to enumerate all feasible solutions, and techniques for various kinds of enumeration problems are discussed \cite{wasa2016enumeration}.  Particularly, efficient enumeration methods based on {\em decision diagrams} \cite{Knuth2009:BDD-taocp,Inoue2014:STTT,Kawahara2017:IEICE} have been developed, and for some kinds of practical applications, we can generate a very large number of all feasible solutions in a practical time and space. 

In this paper, we formulate the problems as {\em cost-bounded combinatorial problems}, such that the total cost of items should not be more (or less) than a given cost bound. We propose a fast method for exactly enumerating a very large number of all lower cost solutions. This technique will lead some new prospective applications. For example, our method can be used for calculating the exact ranking for any given solution, and it enables us to compute $p$-values used in statistical assessment. Another good application will be ``quality-controlled'' random sampling, which will be useful for generating a large amount of training data in machine learning. 

Our proposed algorithm is based on backtracking for a given {\em decision diagram} which implicitly represents all feasible solutions. The main idea is to memoize the intervals of cost bounds to avoid duplicate search in the backtracking process. On this type of enumeration problems, we know a classical method with dynamic programming which requires a pseudo-polynomial time with the total cost values; however, the cost values often become large in practical applications when dealing with path lengths, financial incomes, national population, etc., and in such cases it is not easy to execute. In contrast, the computation time of our method does not directly depend on the total cost values, but is bounded by the input and output size of the decision diagrams. Therefore, it can be much faster than the existing methods if the cost values are large but the input/output of the decision diagrams are well-compressed. Our experimental results show that, for several nontrivial size instances of the Hamiltonian path problem, we succeeded in exactly enumerating billions of all lower cost solutions in a few seconds, which is hundred or more times faster than existing methods. Our method can be regarded as a novel search algorithm which integrates the two classical techniques, branch-and-bound and dynamic programming.

In the rest of this paper, we briefly review decision diagram-based techniques for solving combinatorial problems in Section 2, and then discuss the motivation and technical difficulties of enumeration problems in Section 3. We present our algorithm and discuss the time and space complexity in Section 4. Experimental evaluations are shown in Section 5, followed by related work and summary. 

\section{Preliminaries}
\subsection{Cost-Bounded Combinatorial Problems}
In this paper, we define a set of $n$ items $I = \{1, 2, \ldots, n \}$, and we call any subset $X \subseteq I$ an {\em item combination}. For a given {\em constraint function} $f : 2^I \rightarrow \{0, 1\}$, the set of {\em feasible solutions} for $f$ is denoted by $\mathcal{S}_f = \{X \subseteq I  \mid f(X) = 1\}$. For each item $i \ \ (1\leq i \leq n)$, an integer-valued cost $c_i \in \mathbb{Z}$ is defined. The cost of $X$ is defined as $\displaystyle \mathrm{Cost}(X) = \sum_{i \in X} c_i$. A combinatorial optimization problem can be formulated as a problem to find a feasible solution $X^* \in \mathcal{S}_f$ which minimizes (or maximizes) $\mathrm{Cost}(X^*)$.

In this paper, we discuss the enumeration problems derived from the optimization problems. Specifying a cost bound 
$b \in \mathbb{Z}$, we define a problem to find a solution $X \in \mathcal{S}_f$ satisfying $\mathrm{Cost}(X) \leq b$. We call such problems {\em cost-bounded combinatorial problems} for $(f, b)$. We consider a fast method for exactly enumerating all solutions of this type of problems.

\begin{figure}[t]
\begin{center}
\includegraphics[width=75mm]{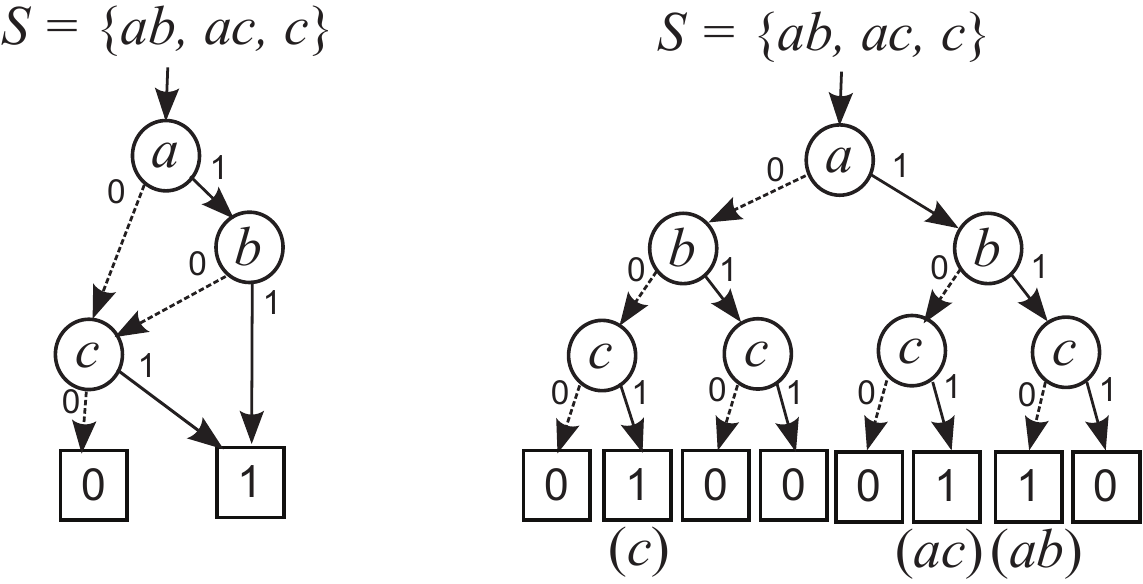}\\
\hspace*{9mm} (a) ZDD \hspace{18mm} (b) Binary decision tree \hspace*{8mm}
\caption{ZDD and binary decision tree.}
\label{fig:ZDD-Tree}
%\end{center}
%\end{figure}
\vspace{5mm}
%\begin{figure}[t]
%\begin{center}
\includegraphics[width=95mm]{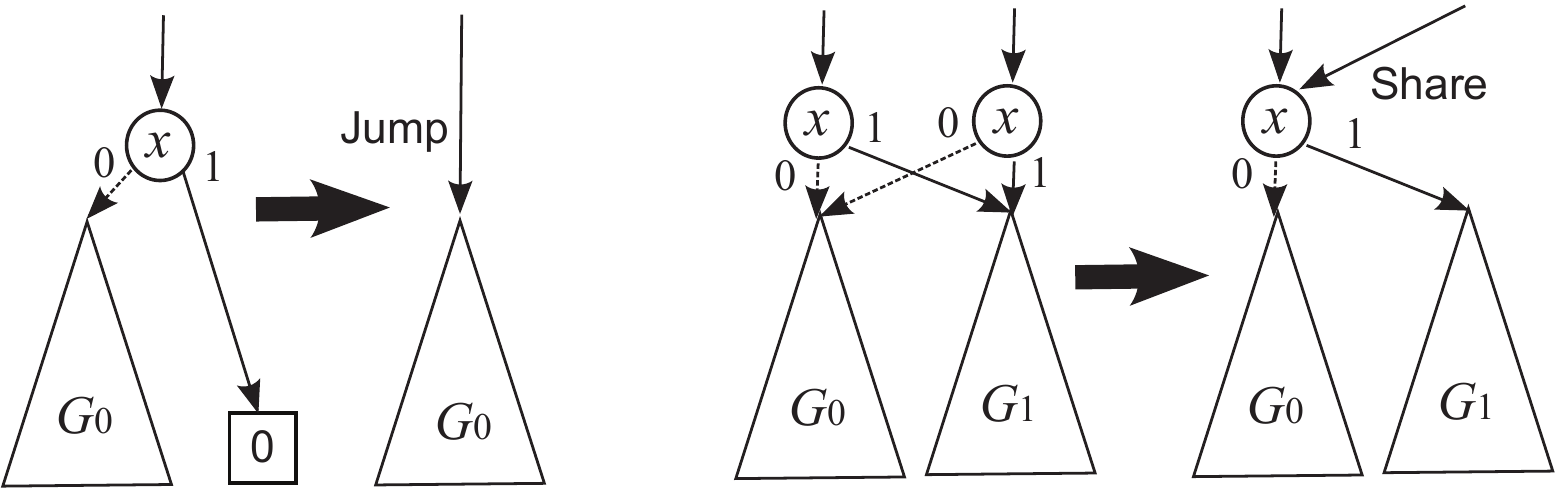}\\
\hspace*{2mm} (a) Node deletion \hspace{23mm} (b) Node sharing \hspace*{10mm}
\caption{ZDD reduction rules}
\label{fig:ZDDrules}
\vspace{-5mm}
\end{center}
\end{figure}

\subsection{ZDDs and Combinatorial Optimization}
A {\em zero-suppressed binary decision diagram} (ZDD) \cite{Minato93} is a graphical representation for a set of item combinations\footnote{In this paper, we assume to use ZDDs but original BDDs (binary decision diagrams) can be used for our algorithm almost similarly.}.
As illustrated in Fig.~\ref{fig:ZDD-Tree}, it is derived by reducing a binary decision tree, which indexes the set of item combinations. In this reduction process for ZDDs, we apply the following two reduction rules:
\begin{itemize}
\item Delete all nodes whose 1-edge directly points to the 0-terminal node. (Fig.~\ref{fig:ZDDrules}(a))
\item Share all equivalent nodes having the same item and the same pair of children. (Fig.~\ref{fig:ZDDrules}(b))
\end{itemize}
If we apply the two reduction rules as much as possible, we obtain a canonical form for a given set of combinations. The compression ratio of a ZDD depends on instance, but it can be 10 or 100 times compact in some practical cases. In addition, we can systematically construct a ZDD by a set operation (i.e., {\em union},  {\em intersection} and {\em set difference}) for a given pair of ZDDs. This algorithm is based on a recursive procedure with hash table techniques presented by Bryant \cite{Bryant86}. The computation time is theoretically bounded by the product of the two operands' ZDD sizes; however, empirically in many cases, it is almost linear time for the sum of input and output ZDD sizes.

Using ZDDs, we can solve combinatorial optimization problems in a simple procedure. First we describe a Boolean constraint $f$ in a formula of set operations. Then we can construct a ZDD for $f$ by applying ZDD set operations according to the formula. In this ZDD, each path from the root node to the 1-terminal node (we call it a {\em 1-path}) corresponds to one feasible solution in $\mathcal{S}_f$. The optimal solution can be found by exploring a 1-path having the minimum (or maximum) cost. A naive depth-first search takes an exponential time, but we can find the optimal solution in a linear time for the input ZDD size $O(|f|)$ by memoizing the optimal cost of the 1-paths starting from each ZDD node to the 1-terminal node, because we don't have to traverse the same ZDD node more than once by pruning the depth-first search using the memoized cost value. Thus, it is not difficult to obtain one optimal solution if we have constructed a ZDD for $f$ in main memory. 

Now, our interest is how to enumerate not only the optimal one but all the solutions no more than a given cost bound. 

\section{Enumeration of All Lower Cost Solutions}
\subsection{Motivating Example}
Knuth presented a graph instance of continental US map (Fig.~\ref{fig:US48}) in TAOCP, Chapter 7.1.4, Vol.~4-1 \cite{Knuth2009:BDD-taocp}. He discussed a problem to enumerate all Hamiltonian paths (visiting every state exactly once) across the US continent. This graph consists of 48 vertices and 105 edges, and every edge has a cost of mileage between the two state capital cities. In this Book, Knuth presented {\em simpath} algorithm to construct a ZDD representing a set of all feasible solutions of this problem. The algorithm is a kind of dynamic programming method by scanning the graph with a crossing section, called {\em frontier}. This algorithm constructs a ZDD in a top-down breadth-first manner according to the frontier scanning. This method can deal with many kinds of edge-induced graph subsets, such as simple and Hamiltonian paths, cycles, trees, connected components, cut sets, etc. We generically call such ZDD construction method {\em Frontier-based methods} \cite{Kawahara2017:IEICE} and working on various applications in practical problems, such as loss minimization of a power grid network~\cite{Inoue2014:IEEE-TSG}, hotspot detection~\cite{Ishioka2019}, and electoral redistricting~\cite{kawahara2017partition,fifield2020election}.

\begin{figure}[t]
\begin{center}
\hspace*{-5mm}
\includegraphics[width=100mm]{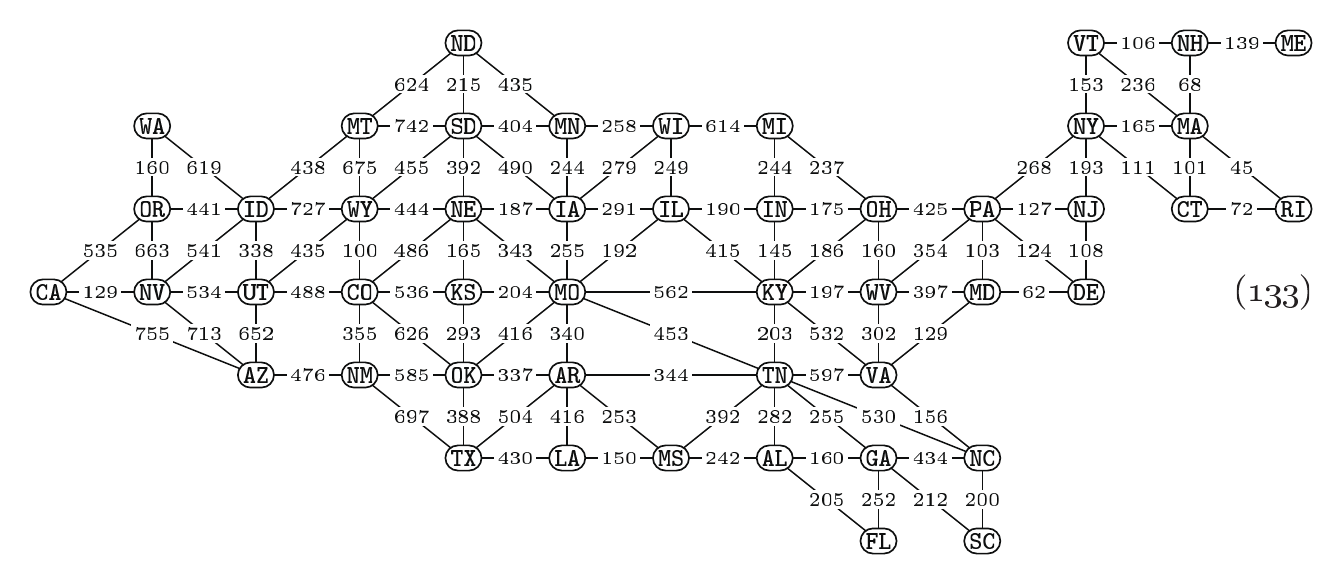}
\caption{US map graph shown in the Knuth's book \cite{Knuth2009:BDD-taocp}}
\label{fig:US48}
\vspace{-5mm}
\end{center}
\end{figure}

Using Knuth's algorithm, we can construct a ZDD enumerating all 6,876,928 Hamiltonian paths from WA to ME. The ZDD size is only 3,616 nodes, and computation time is only 0.02 sec. Knuth's method is surprisingly efficient to enumerate all feasible solutions; however, it is still difficult to enumerate a subset of solutions bounded by the cost function, such that the lengths are no more than 20\% increase from the shortest one. In real-life situations, we may accept a certain level of cost increase but want to trade another ``soft'' constraint, for example, we would like to avoid a route expecting a heavy snow in winter. In other cases, for p-value calculation in statistical analysis, we need to know (a big number of) the ranking of an observed solution.  Another useful application would be ``quality-controlled'' random sampling, which generates random solutions in a nearly optimal range. We expect that it is a practically useful and important task to enumerate all the cost-bounded solutions.

\subsection{Conventional Top-$k$ Search Methods}
{\em Top-$k$ search} (or {\em $k$-best search}) is a conventional approach to generate top-$k$ solutions based on existing optimization algorithms for finding one best solution. For example, there is a classical algorithm to solve a practical size of traveling salesman (= shortest Hamiltonian cycle) problem using a linear programming solver \cite{Miller1960}. This algorithm can find one minimum-cost solution, and we can generate top-$k$ solutions by $k$-times repetition of the algorithm with additional constraints to excluding the solutions already found. This approach requires a large computation time when $k$ becomes large. For Knuth's US map problem, we have 16,180 solutions up to 10\% cost increase, and 939,209 solutions up to 20\% increase. In such cases, the conventional top-$k$ approach is too time consuming. If we consider only the simple (not Hamiltonian) top-$k$ shortest paths problem, there is a sophisticated algorithm \cite{Eppstein98} using Dijkstra method with a heap data structure, but it is still time consuming when $k$ becomes a million or more.

\subsection{Classical Method with Dynamic Programming}
On this type of enumeration problems, we know a classical method with dynamic programming using a DP table to store the subtotal of the costs for each decision, which requires a pseudo-polynomial time and space with the total cost values. However, the cost values often become large in practical applications when dealing with path lengths, financial incomes, national populations, etc. For Knuth's US map problem, the total cost value becomes 35,461 (miles), and the DP table may have three millions of cells. This could work in a modern computer but time-consuming to explore all Hamiltonian paths with this DP table.

\begin{figure}[t]
\begin{center}
\includegraphics[width=110mm]{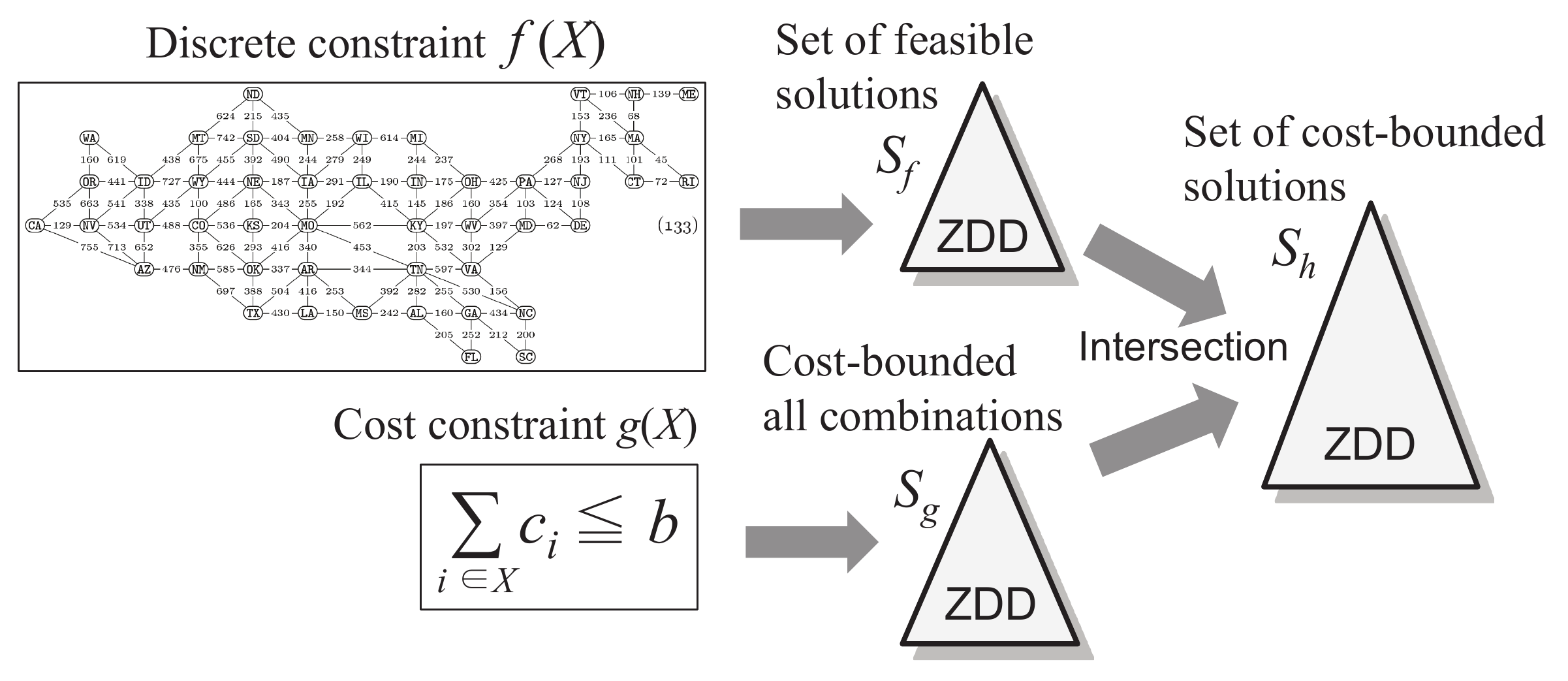}
\caption{Conventional method using ZDD operations }
\label{fig:ZDDintsec}
\vspace{-5mm}
\end{center}
\end{figure}

\begin{table}[t]
\begin{center}
\caption{ZDD sizes for $\mathcal{S}_g$ with uniform random costs}
\label{tab:cost}
\renewcommand{\arraystretch}{0.8}
{\small 
\begin{tabular}{|r|r|r|r|}
\hline
$n$ & cost = 1 & rand [1, 100] & rand [1, 10000] \\
\hline
10 & 30 & 53 & 52 \\
20 & 110 & 1,645 & 2,128 \\
30 & 240 & 7,225 & 54,302 \\
40 & 420 & 16,124 & 546,035 \\
50 & 650 & 27,589 & 1,501,267 \\
60 & 930 & 39,931 & 2,729,644 \\
70 & 1,260 & 55,387 & 4,193,333 \\
80 & 1,640 & 76,355 & 6,033,872 \\
90 & 2,070 & 102,831 & 8,394,357 \\
100 & 2,550 & 123,131 & 10,484,185 \\
110 & 3,080 & 148,508 & 12,933,163 \\
120 & 3,660 & 173,369 & 16,219,732 \\
\hline
\end{tabular} 
}
\vspace{-5mm}
\end{center}
\end{table}

\subsection{Conventional ZDD-Based Method}
ZDD-based method may be effective for representing a large number of combinations. Figure~\ref{fig:ZDDintsec} illustrates the existing method for enumerating all cost-bounded solutions using conventional set operations of ZDDs. We construct a ZDD for feasible solutions $f$, and also construct another ZDD $g$ for the cost constraint $\mathcal{S}_g = \{X \subseteq I \mid \mathrm{Cost}(X) \leq b \}$. We then apply an intersection operation between the two ZDDs, and the result of ZDD $h$ represents the set of all cost-bounded solutions $\mathcal{S}_h = \{X \in \mathcal{S}_f \mid \mathrm{Cost}(X) \leq b \}$. 

This method seems correct and effective, but unfortunately, there is a significant weak point. The ZDD for $\mathcal{S}_g$ has essentially the same structure as the DP table used in classical dynamic programming. A linear cost constraint function generates a small ZDD only if the total cost value is bounded by a small constant, but in general it is exponential \cite{Hosaka1997}\footnote{They discuss the complexity on original BDDs, but essentially similar for ZDDs. }. Table~\ref{tab:cost} shows our preliminary experiments to see the ZDD sizes for the $n$-input cost constraints, where each cost is a uniformly random in [1, 100] or [1, 10000], with the cost bound $b$ as a half of the average cost. We can observe that the ZDD size for $\mathcal{S}_g$ grows large more than a million nodes when the range of cost values reaches thousands and the number of items are fifty or more. In those cases, the conventional ZDD-based method are not efficient, even if ZDDs for $\mathcal{S}_f$ and $\mathcal{S}_h$ are small. In the following sections, we present a new method for constructing ZDD for $\mathcal{S}_h$ without using ZDD for $\mathcal{S}_g$.

\section{Proposed Algorithm}
\subsection{Basic Procedure and Conventional Memoizing}
Our algorithm assumes the following inputs: the number of items $n$, a ZDD $f$ for the set of the feasible solutions $\mathcal{S}_f$, cost values $c_i \ \ (1\leq i \leq n)$, and a cost bound $b$. The cost values and the bound are integers and can be zero or negative. We do not have to fix the bit-width of the integers. The output of the algorithm is a newly constructed ZDD $h$ for $\mathcal{S}_h = \{X \in \mathcal{S}_f \mid \mathrm{Cost}(X) \leq b \}$. Here is our basic procedure {\bf BacktrackNaive}.

%{\small
\begin{spacing}{0.83}
\noindent\hrulefill\\
1: \ $\mathbf{ZDD}$ {\bf BacktrackNaive}$(\mathbf{ZDD} \ f, \mathbf{int} \ b)$\\
2: \ \{\\
3: \ \ \ \ \ $\mathbf{if}$ $f = [0]$ $\mathbf{return}$ $[0]$ \ /* 0-terminal case */\\
4: \ \ \ \ \ $\mathbf{if}$ $f = [1]$ $\mathbf{then}$ \hspace{7.5mm} /* 1-terminal case */\\ 
5: \ \ \ \ \ \ \ \ \ $\mathbf{if}$ $b \geq 0$ $\mathbf{return}$ $[1]$\\
6: \ \ \ \ \ \ \ \ \ $\mathbf{else}$ $\mathbf{return}$ $[0]$\\
7: \ \ \ \ \ /* $f$ consists of $(x, f_0, f_1)$ */\\
8: \ \ \ \ \ $h_0 \leftarrow$ {\bf BacktrackNaive}$(f_0, b)$\\
9: \ \ \ \ \ $h_1 \leftarrow$ {\bf BacktrackNaive}$(f_1, b-\mathrm{cost}(x))$\\
10:\ \ \ \ \ $h \leftarrow \mathbf{ZDD}(x, h_0, h_1)$ /* applying ZDD reduction rules */\\
11:\ \ \ \ \ $\mathbf{return}$ $h$\\
12:\ \}\\
\hspace*{0mm}\hrulefill
\end{spacing}
%}
This procedure recursively performs a simple backtrack traversal on the input ZDD $f$ in a depth-first manner. On each recursive step for the problem $(f, b)$, it is divided into the two sub-problems $(f_0, b)$ and $(f_1, b-\mathrm{cost}(x))$. Repeating this recursive call, we eventually reach a terminal node. If it is 1-terminal and the current cost bound is not negative, then we accept it as a solution and return the 1-terminal node before backtracking. Otherwise we reject it with the 0-terminal node. In the case of non-terminal step, we obtain $h_0$ and $h_1$ as the results of the two sub-problems, and we generate an output ZDD node $h$. Note that $\mathbf{ZDD}(x, h_0, h_1)$ in line 10 means a subroutine to generate a new ZDD node specified by the three attributes. If an equivalent ZDD node already exists, then just a pointer to the existing node is returned to avoid generating a redundant node.

\begin{figure}[t]
\begin{center}
\includegraphics[width=52mm]{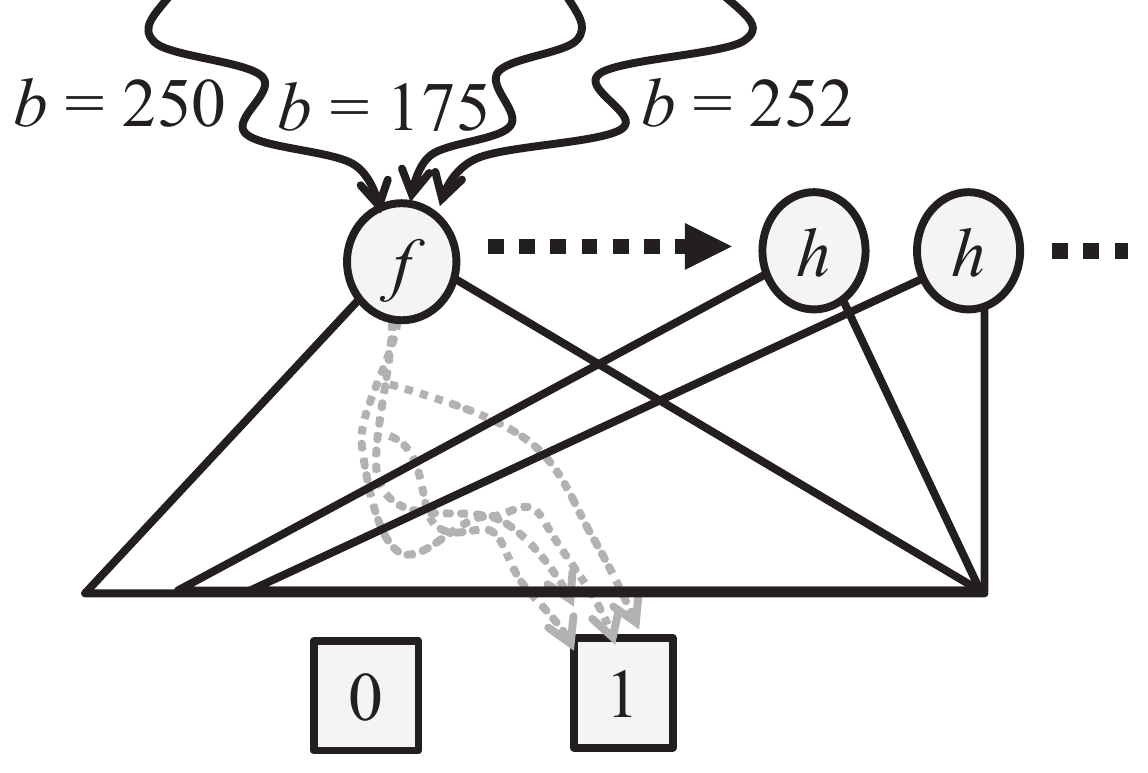}
\caption{Revisiting the same ZDD node with different cost bounds}
\label{fig:backtrack}
\vspace{-5mm}
\end{center}
\end{figure}

{\bf BacktrackNaive} can correctly construct a ZDD for $S_h$; however, this basic algorithm may visit the same ZDD node many times, and always requires an exponential number of steps for the depth-first traversal. In many ZDD set operations, we can use a memoizing technique (also called {\em operation cache}) to avoid duplicated traversal, and usually it is very effective to accelerate the operations. However, as illustrated in Fig.~\ref{fig:backtrack}, it is not easy for the cost-bounded ZDD operations because when visiting the same ZDD node at the second time, the subtotal cost of used items from the root node to the current node may be different from one at the first visit. In such cases, the operation result may not be the same. Thus, we should check a pair of $(f, b)$ as a key to the memo.

Here is a procedure {\bf BacktrackMemo} with a conventional memoizing technique.

%{\small
\begin{spacing}{0.83}
\noindent\hrulefill\\
1: \ $\mathbf{ZDD}$ {\bf BacktrackMemo}$(\mathbf{ZDD} \ f, \mathbf{int} \ b)$\\
2: \ \{\\
3: \ \ \ \ \ $\mathbf{if}$ $f = [0]$ $\mathbf{return}$ $[0]$ \ /* 0-terminal case */\\
4: \ \ \ \ \ $\mathbf{if}$ $f = [1]$ $\mathbf{then}$ \hspace{7.5mm} /* 1-terminal case */\\ 
5: \ \ \ \ \ \ \ \ \ $\mathbf{if}$ $b \geq 0$ $\mathbf{return}$ $[1]$\\
6: \ \ \ \ \ \ \ \ \ $\mathbf{else}$ $\mathbf{return}$ $[0]$\\
7: \ \ \ \ \ $h \leftarrow$ {\bf memo}$[f, b]$ ; $\mathbf{if}$ $h$ exists $\mathbf{return}$ $h$\\
8: \ \ \ \ \ /* $f$ consists of $(x, f_0, f_1)$ */\\
9: \ \ \ \ \ $h_0 \leftarrow$ {\bf BacktrackMemo}$(f_0, b)$\\
10:\ \ \ \ \ $h_1 \leftarrow$ {\bf BacktrackMemo}$(f_1, b-\mathrm{cost}(x))$\\
11:\ \ \ \ \ $h \leftarrow \mathbf{ZDD}(x, h_0, h_1)$ /* applying ZDD reduction rules */\\
12:\ \ \ \ \ {\bf memo}$[f, b] \leftarrow h$\\
13:\ \ \ \ \ $\mathbf{return}$ $h$\\
14:\ \}\\
\hspace*{0mm}\hrulefill
\end{spacing}
%}

In this pseudo code, $\mathbf{memo}[f, b]$ in line 7 and 12 is the memoizing function, added to avoid duplicated traversal. In this procedure, all possible variations of subtotal costs may be produced and stored in the memo, which is essentially a similar task to make a DP table in dynamic programming method. Unfortunately, if the cost values become large and have a nearly random distribution, the probability of hitting the same subtotal cost is significantly low, and this memoizing technique is not very effective.

\subsection{Interval-Memoized Backtracking}
In the depth-first traversal shown in Fig.~\ref{fig:backtrack}, if we revisit a same ZDD node $f$ with a cost bound $b'$ different from the first bound $b$, the result $h$ may not be the same, but intuitively, if $b$ and $b'$ are very close, the result is likely the same with a good possibility. More formally, the result $h$ must be the same if there is no solution with a cost between $b$ and $b'$. This is the key idea of our proposed method. Here we define the two important functions:
\begin{itemize}
\item $\mathit{accept\_worst}(f, b)$:\\
the worst (highest) cost of an accepted solution.
\item $\mathit{reject\_best}(f, b)$:\\
the best (lowest) cost of rejected one but feasible in $f$.
\end{itemize}
Using those two functions, we can guarantee to compute the same result at the second visit if and only if
$\mathit{accept\_worst}(f, b) \leq b' < \mathit{reject\_best}(f, b)$.

\begin{figure}[t]
\begin{center}
\includegraphics[width=70mm]{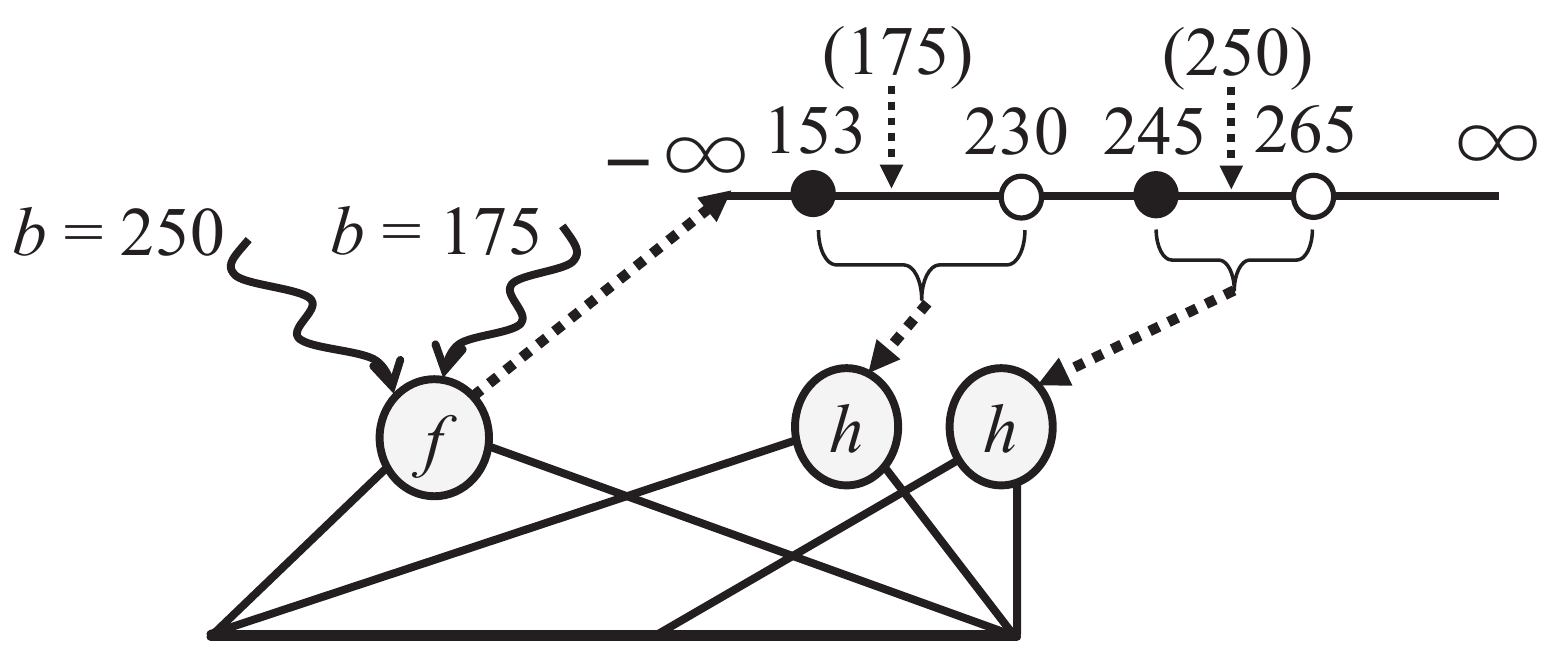}
\caption{Idea of interval-memoizing}
\label{fig:memo}
\vspace{-5mm}
\end{center}
\end{figure}

Figure~\ref{fig:memo} illustrates the idea of interval memoized backtracking. For each ZDD node $f$, we prepare a {\em numerical-ordered memory} to store the intervals. A black dot indicates a point of $\mathit{accept\_worst}$, and white dot is a point of $\mathit{reject\_best}$. In this example, if we first visit $f$ with $b=250$, an effective interval [245, 265) is determined and a pointer to the result $h$ is stored with the interval information. At the next time, if we revisit $f$ with $b=252$, we can surely avoid a duplicate recursive call. On the other hand, if we revisit $f$ with $b=175$, it is out of scope and the depth-first traversal goes on, and then another result is stored with the interval [153, 230).

We can implement the numerical-ordered memories using {\em self-balancing binary search trees}, available as std::map in gnu C++ standard library. It requires $O(\log m)$ time for each read/write in average, where $m$ is the number of entries in the memory. This is not a constant factor but acceptable overhead if the ratio of the ZDD size $|h| / |f|$ is not very large.

Another important issue is how to know the correct intervals. Fortunately, we can easily compute $\mathit{accept\_worst}$ and $\mathit{reject\_best}$ with a simple modification of the backtracking process. Here we present the pseudo code of our algorithm {\bf BacktrackIntervalMemo}.

%{\small
\begin{spacing}{0.83}
\noindent\hrulefill\\
1: \ {\bf BacktrackIntervalMemo}$(\mathbf{ZDD} \ f, \mathbf{int} \ b)$\\
2: \ /* returns a triple $(\mathbf{ZDD} \ h, \mathbf{int} \ \mathit{accept\_worst}, \mathit{reject\_best})$  */\\
3: \ \{\\
4: \ \ \ \ \ $\mathbf{if}$ $f = [0]$ $\mathbf{return}$ $([0], -\infty, \infty)$\\
5: \ \ \ \ \ $\mathbf{if}$ $f = [1]$ $\mathbf{then}$\\ 
6: \ \ \ \ \ \ \ \ \ $\mathbf{if}$ $b \geq 0$ $\mathbf{return}$ $([1], 0, \infty)$\\
7: \ \ \ \ \ \ \ \ \ $\mathbf{else}$ $\mathbf{return}$ $([0], -\infty, 0)$\\
8: \ \ \ \ \ $(h, aw, rb) \leftarrow$ {\bf memo}$_\mathbf{ord}[f, b]$; $\mathbf{if}$ $\mathrm{exists}$ $\mathbf{return}$ $(h, aw, rb)$\\
9: \ \ \ \ \ /* $f$ consists of $(x, f_0, f_1)$ */\\
10:\ \ \ \ $(h_0, aw_0, rb_0) \leftarrow$ {\bf BacktrackIntervalMemo}$(f_0, b)$\\
11:\ \ \ \ \ $(h_1, aw_1, rb_1) \leftarrow$ {\bf BacktrackIntervalMemo}$(f_1, b-\mathrm{cost}(x))$\\
12:\ \ \ \ \ $h \leftarrow \mathbf{ZDD}(x, h_0, h_1)$ /* applying ZDD reduction rules */\\
13:\ \ \ \ \ $aw \leftarrow \mathbf{max}(aw_0, \ aw_1+ \mathrm{cost}(x))$\\
14:\ \ \ \ \ $rb \leftarrow \mathbf{min}(rb_0, \ rb_1+ \mathrm{cost}(x))$\\
15:\ \ \ \ \ {\bf memo}$_\mathbf{ord}[f, [aw, rb)] \leftarrow h$\\
16:\ \ \ \ \ $\mathbf{return}$ $(h, aw, rb)$\\
17:\ \}\\
\hspace*{0mm}\hrulefill
\end{spacing}
%}

This procedure is extended to return not only the ZDD $h$ but also the two integers $aw \ (\mathit{accept\_worst})$ and $rb \ (\mathit{reject\_best})$. The return values on the terminal cases are written in line 4, 6, and 7. If $f$ is a 1-terminal and $b$ is not negative, then we get $aw \leftarrow 0$ since $b=0$ is the lowest value to return the same decision. Similarly, if $b<0$, then we get $rb \leftarrow 0$ since  $b=0$ is the next lowest value to return a different decision. Line 13 and 14 specify the non-terminal case. We already have the intervals $[aw_0, rb_0)$ and $[aw_1, rb_1)$ returned from the two sub-problems. Then we choose $aw$ for the maximum one, and $rb$ for the minimum one of the two cases. In this way, after generating each ZDD node $h$, we already know the exact interval $[aw, rb)$ to be stored, only with a constant factor of overhead.

Our algorithm has some interesting properties. If we give $b=-\infty$, the algorithm always returns empty solution and $\mathit{reject\_best}$ shows the minimum cost of the solution in $f$. Similarly, for $b=\infty$, it always returns $f$ itself and $\mathit{accept\_worst}$ shows the maximum cost. For such two cases, the behavior of the algorithm is essentially same as a branch-and-bound for solving optimization problems, and the computation time is linearly bounded by the input size $O(|f|)$. In addition, if we give an arbitrary bound $b$, our algorithm efficiently constructs a ZDD which essentially corresponds to a DP table in a dynamic programming. Thus, we can see that our method integrates the two classical techniques, branch-and-bound and dynamic programming, and the same program can be used both for optimization and enumeration problems.

The interval-memoizing technique has another advantage that it is effective for multiple times of executions. We may repeat the backtracking algorithm many times for the same $f$ with different cost bound to analyze the distribution of the solutions.
In such cases, the results of different ZDDs may share their subgraphs, and most of the shared parts are immediately returned from the memo after the second time. 

After generating ZDDs for the cost-constraint problems, we can utilize various kinds of algebraic set operations of ZDD library. For example, if we already have the two ZDDs $h_{lb}$ and $h_{ub}$ with the cost bounds $lb$ and $ub$, then the set difference $(h_{ub} \setminus h_{ub})$ gives all solutions in a cost range $(lb, ub]$.

\subsection{Complexity Analysis}
Here we discuss the complexity of our algorithm. Intuitively, the computation time is bounded by the input and output ZDD sizes.
\begin{theo}
When the input ZDD $f$ represents a power set $2^I$, the time complexity of {\bf BacktrackIntervalMemo} is bounded by $O(|f| + |h| \log |h|)$, where $|f|$ and $|h|$ are the input and output ZDD sizes.
\label{th4}
\end{theo}
\begin{proof}
When revisiting a ZDD node in $f$ with some different $b$, the interval-memo never fail to detect whether the outputs are the same or not. Therefore, we never perform duplicated traversals in the output ZDD, and thus the total accesses to the memo is bounded by $O(|h|)$. We know that each access to the memo requires $O(\log m)$ steps, where $m$ is total entries bounded by $|h|$. On the other hand, even if $h$ is an empty set or much smaller than $f$, the algorithm visits all the input ZDD nodes once for each, so at least $O(|f|)$ steps are required. Thus, we conclude that the time complexity is $O(|f| + |h| \log |h|)$.
\qedhere
\end{proof}

The above theorem discusses only in a basic case when $f$ is a power set, where the problem has only a cost constraint. Unfortunately, this theorem does not always stand for all general cases of $f$. It is correct that we never perform duplicated traversals in the output ZDD $h$ when we revisiting a same ZDD node in $f$. However, there is the case that different two ZDD nodes in $f$ may produce a same ZDD node in $h$. In such cases the memo cannot predict it and we may traverse the same subgraph of $h$ multiple times. We can artificially design such a counterexample case, but it does not likely happen in practical problems. We consider that the computation time is empirically $O(|f| + |h| \log |h|)$ in many practical applications.

\section{Experimental Results}
We implemented our algorithm with a BDD/ZDD package included in {\em Graphillion} library \cite{Graphillion2013}. The program is written in C and C++ code. We used a PC (Intel Xeon W-2225, 4 cores 4.1GHz, 256GB main memory), Linux - Ubuntu 20.04 and g++ 9.3.0. 

\begin{table}[t]
\begin{center}
\caption{Experimental results for Knuth's US map problem}
\label{tab:US48}
\renewcommand{\arraystretch}{0.8}
\setlength{\tabcolsep}{0.8mm} 
{\small 
Continental US map graph ($V$: 48, $E$: 105)\\
\begin{tabular}{|r|r|r|r|r||r|r||r|}
\hline
bound $b$ (ratio) & \#solutions & ZDD $|h|$ & \multicolumn{2}{c||}{proposed method} & \multicolumn{2}{c||}{conventional memo} & clingo \\
\cline{4-8}
& & & time(sec) & \#calls & time(sec) & \#calls & time(sec) \\
\hline
11,698 (1.00) & 1 & 47 & 0.028 & 7,329 & 0.066 & 197,891 & 10.784 \\
11,814 (1.01) & 8 & 99 & 0.029 & 7,441 & 0.070 & 197,891 & 5.243 \\
11,931 (1.02) & 28 & 152 & 0.028 & 7,555 & 0.069 &197,891 & 7.028 \\
12,282 (1.05) & 388 & 1,001 & 0.029 & 9,489 & 0.078 & 229,257 & 8.783 \\
12,867 (1.10) & 16,180 & 9,679 & 0.035 & 28,127 & 0.118 & 340,277 & 12.080 \\
14,037 (1.20) & 939,209 & 72,808 & 0.089 & 155,985 & 0.664 & 1,573,161 & 26.276 \\
15,207 (1.30) & 4,525,541 & 99,759 & 0.113 & 206,089 & 1.414 & 3,180,475 & 40.463 \\
16,377 (1.40) & 6,702,964 & 38,548 & 0.051 & 78,295 & 1.743 & 4,062,971 & 39.015 \\
17,547 (1.50) & 6,876,526 & 4,934 & 0.029 & 9,971 & 1.784 & 4,274,085 & 36.879 \\
(*) 18,040 (1.54) & 6,876,928 & 3,616 & 0.028 &7,233 & 1.832 & 4,281,461 & 37.031 \\
\hline
\multicolumn{8}{l}{(*): maximum cost. (here $h = f$)}\\
\end{tabular} \\
}
\vspace{-5mm}
\end{center}
\end{table}

\subsection{Results for Hamiltonian Path Problem}

Here we show our experimental results for nontrivial size instances of Hamiltonian path problem, which is one of the most popular and hard combinatorial problems.
First, we show the results for Knuth's US map problem (Hamiltonian paths from WA to ME), discussed in Section 3. Here, the number of the items $n = 105$, the total feasible solutions (without cost bound) $|S_f| = 6,876,928$, the input ZDD size $|f|$ is 3,616 nodes, the minimum (shortest) cost is 11,698, and the maximum (longest) cost is 18,040. We then applied our algorithm with various cost bounds. In Table~\ref{tab:US48}, the computation time includes the time for constructing the input ZDD $f$, the time for constructing the output ZDD $h$, and the time for counting the output ZDD size $|h|$. The column ``proposed method'' shows the execution time and the total recursive calls of our algorithm using interval-memoizing, and ``conventional memo'' compares the execution time and the total calls when using the conventional memoizing technique discussed in Section 4.1. We also show the results of ``clingo'', the answer set programming system {\em clingo} \cite{Gebser2012} ver. 5.4.0. This is one of the efficient enumeration tools for combinatorial problems,  based on SAT-solving and model counting techniques\footnote{CPLEX and gurobi also have an option to output multiple solutions \cite{Danna2007}; however, this ``solution pool" feature finds more solutions directly from intermediate feasible solutions, and it is too time consuming when the number of the solutions becomes very large. To our knowledge, clingo is currently one of the most powerful tools for counting thousands or more solutions.}.

The experimental results show that we succeeded in exactly enumerating millions of cost-bounded Hamiltonian paths for a real-life instance only in 0.1 sec. To see the number of recursive calls, we can observe that the computation time empirically follows $O(|f| + |h| \log |h|)$ as discussed in previous section. We can also confirm that our proposed method is very effective especially when $|h|$ is well-compressed. Interval-memoizing is ten times to six hundred times faster than using the conventional memoizing. The existing tool clingo also succeeded in counting the exact number of solutions as well, but it is hundred or more times slower than our proposed method.

\begin{table}[t]
\begin{center}
\caption{Experimental result for Hamiltonian paths in grid graphs}
\label{tab:nxn}
\renewcommand{\arraystretch}{0.8}
\setlength{\tabcolsep}{0.6mm} 
{\small 
8 $\times$ 8 grid graph ($V$: 81, $E$: 144)\\
\begin{tabular}{|r|r|r|r|r||r|r||r|}
\hline
bound $b$ (ratio) & \#solutions & ZDD $|h|$ & \multicolumn{2}{c||}{proposed method} & \multicolumn{2}{c||}{conventional memo} & clingo \\
\cline{4-8}
& & & time(sec) & \#calls & time(sec) & \#calls & time(sec) \\
\hline
113,552 (1.00)  & 1 & 80 & 0.087 & 93,393 & 114.937 & 151,236,577 & ($>$3,600) \\
114,688 (1.01) & 2,854 & 7,693 & 0.092 & 112,121 & 116.907 & 153,108,163 & ($>$3,600) \\
115,823 (1.02) & 125,880 & 82,363 & 0.165 & 285,709 & 120.376 & 155,502,227 & ($>$3,600) \\
116,959 (1.03) & 2,364,994 & 432,455 & 0.639 & 1,046,119 & 126.568 & 157,794,855 & ($>$3,600) \\
119,230 (1.05) & 133,629,862 & 3,094,302 & 5.523 & 6,473,755 & 138.857 & 163,635,511 & ($>$3,600) \\
122,636 (1.08) & 1,893,484,003 & 5,848,008 & 10.976 & 11,893,469 & 200.010 & 213,488,793 & ($>$3,600) \\
124,907 (1.10) & 2,635,409,530 & 2,101,940 & 3.412 & 4,281,999 & 245.526 & 258,640,783 & ($>$3,600) \\
127,178 (1.12) & 2,687,915,888 & 218,309 & 0.255 & 450,137 & 253.511 & 274,526,635 & ($>$3,600) \\
(*) 130,079 (1.15) & 2,688,307,514 & 46,613 & 0.090 & 93,227 & 247.261 & 275,391,699 & ($>$3,600) \\
\hline
\multicolumn{8}{l}{(*): maximum cost. (here $h = f$)}
\end{tabular} \\
\vspace{2mm}
10 $\times$ 10 grid graph ($V$: 121, $E$: 220)\\
\begin{tabular}{|r|r|r|r|r|}
\hline
bound $b$ (ratio) & \#solutions & ZDD $|h|$ & \multicolumn{2}{c|}{proposed method} \\
\cline{4-5}
& & & time(sec) & \#calls \\
\hline
170,010 (1.00) & 1 & 120 & 0.588 & 997,797\\
%170,350 (1.002) & 126 & 1,864 & 0.588 & 1,000,991\\
%170,860 (1.005) & 4,622 & 17,470 & 0.601 & 1,037,925\\
171,710 (1.01) & 416,589 & 276,180 & 0.896 & 1,641,231\\
%172,560 (1/015) & 14,431,091 & 2,130,927 & 3.922 & 5,761,345\\
173,410 (1.02) & 270,414,340 & 10,388,829 & 20.667 & 23,437,909\\
%174,260 (1.025) & 3,207,577,387 & 353,94,558 & 79.395 & 75,462,571\\
175,110 (1.03) & 26,560,896,936 & 89,730,352 & 219.796 & 186,280,687\\
%175,960 (1.035) & 164,459,408,037 & 179,537,233 & 469.631 & 368,334,553\\
%176,810 (1.04) & 798,958,450,636 & 300,067,017 & 811.112 & 609,938,145\\
178,511 (1.05) & 10,319,390,767,690 & 586,360,102 & 1,684.215 & 1,183,335,939\\
%180,211 (1.06) & 68,115,615,335,404 & 858,208,685 & 2,477.385 & 1,726,708,649\\
%181,911 (1.07) & 259,141,908,047,070 & 1,056,687,840 & 3,180.527 & 2,121,970,599\\
183,611 (1.08) & 623,456,177,103,148 & 1,154,540,999 & 3,411.512 & 2,318,089,817\\
187,011 (1.10) & 1,311,263,635,264,660 & 1,002,804,299 & 2,980.704 & 2,009,425,775\\
190,411 (1.12) & 1,442,845,484,382,530 & 460,708,572 & 1,255.781 & 923,313,563\\
%192,111 (1.13) & 1,445,616,641,527,010 & 190,484,024 & 482.433 & 381,952,009\\
195,512 (1.15) & 1,445,778,909,234,550 & 3,599,172 & 5.565 & 7,224,627\\
(*) 198,385 (1.17) & 1,445,778,936,756,068 & 498,417 & 0.664 & 996,835 \\
\hline
\multicolumn{5}{l}{(*): maximum cost. (here $h = f$)}
\end{tabular} 
}
\vspace{-5mm}
\end{center}
\end{table}

Next, we show the results for another type of graph instance: $n \times n$ grid graphs. This graph consists of $(n+1)^2$ vertices and $2 n(n+1)$ edges, each edge has a uniformly random cost in the range [1000, 1999]. We tried to enumerate all lower cost Hamiltonian paths between the two opposite corners of the grid graph. Table~\ref{tab:nxn} shows the results for $8 \times 8$ and $10\times 10$ graphs. 
In the table of the 10$\times$10 instance, we only shows the results of our proposed method since the other methods are too time-consuming. The experiments show that our method can exactly enumerate up to billions of all lower cost solutions in a few seconds, and up to trillions of solutions in an hour. We can observe that the ratio of ``\#calls'' in proposed method and conventional memoizing method becomes more significant than the results for US map example, as much as several thousand times. Unfortunately, clingo failed
to find any one Hamiltonian path in one hour timeout.

\begin{table}[t]
\begin{center}
\caption{Experimental results for simple path problem with Knuth's US map graph}
\label{tab:US48s}
\renewcommand{\arraystretch}{0.8}
\setlength{\tabcolsep}{0.8mm} 
{\small 
Continental US map graph ($V$: 48, $E$: 105)\\
\begin{tabular}{|r|r|r|r|r||r|r|}
\hline
bound $b$ (ratio) & \#solutions & ZDD $|h|$ & \multicolumn{2}{c||}{proposed method} & \multicolumn{2}{c|}{conventional memo} \\
\cline{4-7}
& & & time(sec) & \#calls & time(sec) & \#calls \\
\hline
3,680 (1.00) & 1 & 13 & 0.036 & 9,379 & 0.045 & 62,239 \\ 
4,048 (1.10) & 2,117 & 562 & 0.036 & 12,501 & 0.057 & 90,545 \\
4,416 (1.20) & 60,197 & 3,270 & 0.044 & 25,181 & 0.086 & 192,667 \\
5,520 (1.50) & 6,883,565 & 55,894 & 0.123 & 209,471 & 0.360 & 865,777 \\
7,360 (2.00) & 708,403,649 & 755,726 & 1.672 & 2,101,933 & 3.472 & 5,188,657 \\
9,200 (2.50) & 17,172,729,980 & 2,648,737 & 6.511 & 6,690,323 & 10.545 & 13,014,219 \\
11,040 (3.00) & 138,284,569,831 & 4,407,834 & 9.952 & 9,879,225 & 16.862 & 20,078,973 \\
12,880 (3.50) & 386,645,930,955 & 4,303,466 & 9.372 & 9,049,021 & 19.873 & 24,864,029 \\
14,720 (4.00) & 479,798,414,556 & 2,057,063 & 3.906 & 4,164,391 & 20.117 & 27,656,725 \\
16,560 (4.50) & 483,362,229,491 & 152,358 & 0.197 & 307,091 & 19.474 & 28,723,351 \\
(*) 18,062 (4.91) & 483,366,193,920 & 4,636 & 0.039 & 9,273 & 19.098 & 28,829,025 \\
\hline
\multicolumn{7}{l}{(*): maximum cost. (here $h = f$)}\\
\end{tabular} \\
}
\vspace{-5mm}
\end{center}
\end{table}

\begin{table}[t]
\begin{center}
\caption{Experimental results for simple path problem with grid graphs}
\label{tab:nxns}
\renewcommand{\arraystretch}{0.8}
\setlength{\tabcolsep}{0.8mm} 
{\small 
6$\times$6 grid graph ($V$: 49, $E$: 84)\\
\begin{tabular}{|r|r|r|r|r||r|r|}
\hline
bound $b$ (ratio) & \#solutions & ZDD $|h|$ & \multicolumn{2}{c||}{proposed method} & \multicolumn{2}{c|}{conventional memo} \\
\cline{4-7}
& & & time(sec) & \#calls & time(sec) & \#calls \\
\hline
15,818 (1.00) & 1 & 12 & 0.038 & 18,933 & 5.389 & 10,145,889 \\
17,400 (1.10) & 44 & 92 & 0.042 & 19,367 & 5.835 & 10,724,097 \\
18,982 (1.20) & 447 & 269 & 0.040 & 20,443 & 6.411 & 11,704,579 \\
23,727 (1.50) & 9,290 & 3,816 & 0.042 & 30,973 & 7.083 & 12,461,055 \\
31,636 (2.00) & 340,703 & 41,322 & 0.083 & 118,697 & 9.985 & 15,990,787 \\
39,545 (2.50) & 6,008,898 & 225,932 & 0.335 & 508,325 & 19.479 & 26,484,267 \\
47,454 (3.00) & 63,239,073 & 697,401 & 1.210 & 1,458,539 & 52.233 & 57,564,039 \\
55,503 (3.50) & 302,233,010 & 1,124,552 & 2.120 & 2,292,043 & 97.376 & 97,825,915 \\
63,272 (4.00) & 534,048,953 & 712,010 & 1.225 & 1,442,907 & 129.568 & 127,805,447 \\
71,181 (4.50) & 575,410,000 & 95,026 & 0.126 & 192,465 & 136.216 & 140,949,659 \\
(*) 77,103 (4.87) & 575,780,564 & 9,430 & 0.042 & 18,861 & 134.397 & 141,861,465 \\
\hline
\multicolumn{7}{l}{(*): maximum cost. (here $h = f$)}\\
\end{tabular} \\
\vspace{2mm}
7 $\times$ 7 grid graph ($V$: 64, $E$: 112)\\
\begin{tabular}{|r|r|r|r|r|}
\hline
bound $b$ (ratio) & \#solutions & ZDD $|h|$ & \multicolumn{2}{c|}{proposed method} \\
\cline{4-5}
& & & time(sec) & \#calls \\
\hline
17,906 (1.00) & 1 & 14 & 0.076 & 67,255 \\
19,697 (1.11) & 116 & 149 & 0.075 & 68,293 \\
21,487 (1.20) & 1,698 & 565 & 0.079 & 70,681 \\
26,859 (1.50) & 64,727 & 11,312 & 0.090 & 107,193 \\
35,812 (2.00) & 4,672,914 & 227,327 & 0.390 & 634,485 \\
44,765 (2.50) & 171,525,844 & 1,920,313 & 3.856 & 4,264,663 \\
53,718 (3.00) & 4,255,053,119 & 10,146,659 & 25.936 & 21,021,793 \\
62,671 (3.50) & 56,032,321,965 & 30,296,461 & 90.832 & 61,413,577 \\
71,624 (4.00) & 299,488,628,619 & 50,034,563 & 152.106 & 100,812,157 \\
80,577 (4.50) & 655,582,656,971 & 45,051,145 & 138.195 & 90,332,797 \\
89,530 (5.00) & 783,461,438,008 & 15,778,153 & 41.511 & 31,627,659 \\
98,483 (5.50) & 789,354,652,325 & 578,468 & 0.841 & 1,161,531 \\
(*) 103,414 (5.78) & 789,360,053,252 & 33,578 & 0.079 & 67,157 \\
\hline
\multicolumn{5}{l}{(*): maximum cost. (here $h = f$)}
\end{tabular} 
}
\vspace{-5mm}
\end{center}
\end{table}

\subsection{Results for Simple Path Problem}

We also show another set of experimental results for simple path (self-avoiding path) problems. Table~\ref{tab:US48s} shows the results for Knuth's US map problem (simple paths from WA to ME), the same graph instance as used in Table~\ref{tab:US48}.  We may compare Tables~\ref{tab:US48} and \ref{tab:US48s}, and we can observe that in the simple path problem, the difference between the shortest path and the longest path is larger than the case of Hamiltonian path problem. In terms of the number of solutions, the simple path problem gives much larger numbers (483,366,193,920) than the case of Hamiltonian path problem. This means that enumerating  simple paths requires more time and space than enumerating Hamilton paths. Anyway, we can observe that our proposed method is much more efficient than using conventional memoizing techniques.

Table~\ref{tab:nxns} shows the results for simple path problem with the 6$\times$6 and 7$\times$7 grid graphs. In this graph, each edge has a uniformly random cost in the range [1000, 1999]. We tried to enumerate all lower cost simple paths between the two opposite corners of the grid graph. The ratio of ``\#calls'' in proposed method and conventional memoizing method becomes more significant than the results for US map example. As well as the case of Hamiltonian path problem, we can observe a drastic progress to efficiently and exactly enumerate more than billions of  all lower cost solutions.

In the above experimental results, we observed a drastic progress to efficiently and exactly enumerate more than billions of  all lower cost solutions, which have never been practicable by repeating execution of an existing optimization method.

\section{Related Work}
The idea of interval-memoizing is so simple and it is not surprising if a similar technique was already used in different problems. It is not easy to determine the really first appearance; however, as one related to decision diagrams, there is a literature by Ab{\'{\i}}o, et al. \cite{Abio2012:PB}. The work is to construct a BDD for representing a Pseudo-Boolean constraint (as a linear cost function), and then convert the BDD structure into a CNF used in a SAT solver to solve combinatorial problems. In this process, they used the interval-memoizing technique for accelerating BDD construction. The basic idea is similar to us, but our algorithm is applicable to any kind of Boolean function $f$, while Ab{\'{\i}}o's method is limited to the Pseudo-Boolean functions. Another important point is that the conventional SAT solvers are customized for exploring any one solution, not designed for enumerate a large number of solutions.

\section{Summary}
In this paper, we presented a fast method for exactly enumerating all lower cost solutions of a combinatorial problem based on ZDD operations. The computation time of the algorithm is empirically $O(|f| + |h| \log |h|)$, where $|f|$ and $|h|$ means the input and output ZDD sizes. Our experimental results show that we can construct a ZDD including billions of all lower cost solutions in a few second. Our method is drastically faster than existing enumeration systems, especially when the input/output ZDDs are well-compressed.

Our algorithm does not use any specific heuristic which is customized for Hamiltonian path problem, and the input is just a ZDD for a set of all feasible solutions. Thus, our algorithm can be used for any kind of combinatorial problems having a Boolean constraint specified by a ZDD. It would be interesting future work to apply our algorithm to various kinds of real-life problems. Our algorithm works well even if positive and negative cost values are mixed. Considering more complicated cost functions will be an interesting future direction.  Our method can be regarded as a novel search algorithm which integrates the two classical techniques, branch-and-bound and dynamic programming. This method would have many applications in various fields, including operations research, data mining, statistical testing, hardware/software system design, etc.

\end{document}